\def \VersionAuthor {}
\ifdefined\VersionAuthor
	\newcommand{\AuthorVersion}[1]{#1}
	\newcommand{\FinalVersion}[1]{}
\else
	\newcommand{\AuthorVersion}[1]{}
	\newcommand{\FinalVersion}[1]{#1}
\fi

\documentclass[a4paper,10pt]{llncs}
\makeatletter
\AtBeginDocument{%
  \@ifpackageloaded{hyperref}
  {\def\@doi#1{\href{https://doi.org/#1}
      {\ttfamily https://doi.org/#1}\egroup}}
  {\def\@doi#1{\ttfamily https://doi.org/#1\egroup}}
  \def\doi{\bgroup\catcode`\_=12\relax\@doi}}
\makeatother

\usepackage[utf8]{inputenc}
\usepackage[english]{babel}
\usepackage[ruled,vlined,linesnumbered]{algorithm2e}
	\SetKwInOut{Input}{input}
	\SetKwInOut{Output}{output}
	\SetKwProg{Fn}{Function}{}{end}

\usepackage{subcaption}
\usepackage{lscape}

\usepackage{paralist} %

\newenvironment{ienumerate}
	{\ifdefined\VersionLong\begin{enumerate}\else\begin{inparaenum}[\itshape i\upshape)]\fi}
	{\ifdefined\VersionLong\end{enumerate}\else\end{inparaenum}\fi}

\newenvironment{oneenumerate}
	{\ifdefined\VersionLong\begin{enumerate}\else\begin{inparaenum}[1)]\fi}
	{\ifdefined\VersionLong\end{enumerate}\else\end{inparaenum}\fi}

\usepackage{amsmath} %
\usepackage{amssymb} %

\usepackage[misc,geometry]{ifsym} %
\usepackage{varwidth}
\ifdefined \VersionLong
	\newcommand{\LongVersion}[1]{#1}
	\newcommand{\ShortVersion}[1]{}
\else
	\newcommand{\LongVersion}[1]{}
	\newcommand{\ShortVersion}[1]{#1}
\fi

\ifdefined\VersionAuthor
	\usepackage[backend=biber,backref=true,style=alphabetic,url=false,doi=true,defernumbers=true,sorting=anyt,maxnames=99]{biblatex} %
	\renewbibmacro*{doi+eprint+url}{%
		\iftoggle{bbx:doi}
			{\color{black!40}\footnotesize\printfield{doi}}
			{}%
		\newunit\newblock
		\iftoggle{bbx:eprint}
			{\usebibmacro{eprint}}
			{}%
		\newunit\newblock
		\iftoggle{bbx:url}
			{\usebibmacro{url+urldate}}
			{}%
	}

\fi
\usepackage[svgnames,table]{xcolor}
\definecolor{darkblue}{rgb}{0, 0, 0.7}

\usepackage[
		pdfauthor={Etienne Andre, Dylan Marinho, Laure Petrucci, Jaco van de Pol},%
		pdftitle={Efficient Convex Zone Merging in Parametric Timed Automata},
		breaklinks  = true,
		colorlinks  = true,
	\ifdefined \VersionWithComments
	\fi
		citecolor   = blue!50!blue,
		linkcolor   = darkblue,
		urlcolor    = blue!50!green,
	]{hyperref}

\usepackage[capitalise,english,nameinlink]{cleveref} %
\crefname{line}{\text{line}}{\text{lines}} %

\usepackage{tikz}
\usetikzlibrary{arrows,automata,shapes,calc}
\tikzstyle{pta}=[auto, ->, >=stealth']
\tikzstyle{PZG}=[auto, ->, >=stealth']
\tikzstyle{every node}=[initial text=]
\tikzstyle{location}=[circle, minimum size=12pt, draw=black, fill=blue!10, inner sep=2pt]
\tikzstyle{invariant}=[draw=black, dotted, inner sep=1pt] %
\tikzstyle{symbstate} = [draw, rectangle, rounded corners]
\tikzstyle{mergingFigure} = [>=stealth', node distance=1.8cm, yscale=.6]

\usepackage{rotating}
\usepackage{multirow,array}
\newcolumntype{M}[1]{>{\centering\arraybackslash}m{#1}}

\newcommand{\rowHeader}{\rowcolor{blue!20}}

\newcommand{\cellBest}{\cellcolor{green!80}\bfseries}
\newcommand{\cellOne}{\cellcolor{green!50}}
\newcommand{\cellTwo}{\cellcolor{green!40}}
\newcommand{\cellThree}{\cellcolor{green!30}}

\newcommand{\cellFive}{\cellcolor{green!10}}

\usepackage{thmtools,thm-restate}
\declaretheorem[name=Proposition]{prop} %
\crefname{prop}{\text{Proposition}}{\text{Propositions}}

\newcommand{\init}{_0}

\newcommand{\A}{\ensuremath{\mathcal{A}}}
\newcommand{\Actions}{\Sigma}
\newcommand{\action}{\ensuremath{a}}
\newcommand{\assign}{\leftarrow}
\newcommand{\BTrue}{\text{true}}
\newcommand{\BFalse}{\text{false}}
\newcommand{\Clock}{\mathbb{X}} %
\newcommand{\ClockCard}{H} %
\newcommand{\clock}{x} %
\newcommand{\clockval}{\mu} %
\newcommand{\ClocksZero}{\vec{0}}
\newcommand{\compOp}{\bowtie}

\newcommand{\edge}{e}
\newcommand{\Edges}{E}
\newcommand{\longuefleche}[1]{\stackrel{#1}{\longrightarrow}}
\newcommand{\longueflecheRel}[1]{\stackrel{#1}{\mapsto}}

\newcommand{\flecheRel}{{\rightarrow}}

\newcommand{\guard}{g}

\newcommand{\invariant}{I}
\newcommand{\K}{K}

\newcommand{\loc}{\ensuremath{\ell}} %
\newcommand{\locinit}{\loc\init}
\newcommand{\Loc}{L} %

\newcommand{\lterm}{\mathit{lt}}
\newcommand{\Param}{\mathbb{P}} %
\newcommand{\param}{p} %
\newcommand{\ParamCard}{M} %
\newcommand{\pval}{v} %
\newcommand{\sinit}{s\init} %
\newcommand{\concstate}{\ensuremath{s}} %
\newcommand{\States}{S} %

\newcommand{\SuccE}{\mathsf{SuccE}}
\newcommand{\timelapse}[1]{#1^\nearrow}

\newcommand{\TTS}{\ensuremath{T}}
\newcommand{\styleSymbStatesSet}[1]{\ensuremath{\mathbf{#1}}}
\newcommand{\Constr}{\ensuremath{\styleSymbStatesSet{C}}}
\newcommand{\symbstate}{\ensuremath{\styleSymbStatesSet{s}}}  \newcommand{\symbstatey}{\ensuremath{\styleSymbStatesSet{y}}} %
\newcommand{\SymbStates}{\ensuremath{\styleSymbStatesSet{S}}} %
\newcommand{\symbstateinit}{\symbstate\init} %
\newcommand{\symbtrans}{\ensuremath{\styleSymbStatesSet{t}}} 
\newcommand{\SymbTransitions}{\Rightarrow} %
\newcommand{\PZG}{\ensuremath{\styleSymbStatesSet{PZG}}} %
\newcommand{\fieldConstr}{\ensuremath{\mathit{constr}}}
\newcommand{\fieldLoc}{\ensuremath{\mathit{loc}}}
\newcommand{\fieldSource}{\ensuremath{\mathit{source}}}
\newcommand{\fieldTarget}{\ensuremath{\mathit{target}}}

\newcommand{\HiLi}[2][yellow]{\begingroup
	\setlength{\fboxsep}{1pt}
	\colorbox{#1}{#2}
	\endgroup}

\newcommand{\stylealgo}[1]{\ensuremath{\mathsf{#1}}}

\newcommand{\mergeSets}{\ensuremath{\mathit{mergeSets}}}
\newcommand{\mergeonestate}{\ensuremath{\mathit{mergeOneState}}}
\newcommand{\ismergeable}{\ensuremath{\mathit{is\_mergeable}}}
\newcommand{\getsiblings}{\ensuremath{\mathit{getSiblings}}}
\newcommand{\looksiblings}{\ensuremath{\styleSymbStatesSet{SiblingCandidates}}}
\newcommand{\candidates}{\ensuremath{\styleSymbStatesSet{Candidates}}}

\newcommand{\bfsbylayer}{\ensuremath{\stylealgo{layerBFS}}}
\newcommand{\visited}{\ensuremath{\styleSymbStatesSet{Visited}}\xspace}
\newcommand{\queue}{\ensuremath{\styleSymbStatesSet{Queue}}\xspace}
\newcommand{\ismerged}{\ensuremath{\mathit{isMerged}}}
\newcommand{\queuenew}{\ensuremath{\styleSymbStatesSet{Q}_\mathit{new}}}

\newcommand{\EFsynth}{\ensuremath{\stylealgo{EFsynth}}}

\newcommand{\resets}{R}
\newcommand{\projectP}[1]{\ensuremath{#1{\downarrow_{\Param}}}}
\newcommand{\reset}[2]{\ensuremath{[#1]_{#2}}}
\newcommand{\valuate}[2]{\ensuremath{#2(#1)}}

\newcommand{\styleBen}[1]{\texttt{#1}}

\newcommand{\namevisited}{\texttt{Visited}}
\newcommand{\namequeue}{\texttt{Queue}}
\newcommand{\nameordered}{\texttt{Ordered}}
\newcommand{\namefly}{\texttt{Otf}}

\newcommand{\nomerge}{\texttt{Nomerge}}
\newcommand{\mergeTwoTwelve}{\texttt{M2.12}}
\newcommand{\mergeRVMr}{\texttt{RVMr}}
\newcommand{\mergeRVM}{\texttt{RVM}}
\newcommand{\mergeRVCr}{\texttt{RVCr}}
\newcommand{\mergeRVC}{\texttt{RVC}}
\newcommand{\mergeRQMr}{\texttt{RQMr}}
\newcommand{\mergeRQM}{\texttt{RQM}}
\newcommand{\mergeRQCr}{\texttt{RQCr}}
\newcommand{\mergeRQC}{\texttt{RQC}}
\newcommand{\mergeROMr}{\texttt{ROMr}}
\newcommand{\mergeROM}{\texttt{ROM}}
\newcommand{\mergeROCr}{\texttt{ROCr}}
\newcommand{\mergeROC}{\texttt{ROC}}
\newcommand{\mergeOVMr}{\texttt{OVMr}}
\newcommand{\mergeOVM}{\texttt{OVM}}
\newcommand{\mergeOVCr}{\texttt{OVCr}}
\newcommand{\mergeOVC}{\texttt{OVC}}
\newcommand{\mergeOQMr}{\texttt{OQMr}}
\newcommand{\mergeOQM}{\texttt{OQM}}
\newcommand{\mergeOQCr}{\texttt{OQCr}}
\newcommand{\mergeOQC}{\texttt{OQC}}
\newcommand{\mergeOOMr}{\texttt{OOMr}}
\newcommand{\mergeOOM}{\texttt{OOM}}
\newcommand{\mergeOOCr}{\texttt{OOCr}}
\newcommand{\mergeOOC}{\texttt{OOC}}

\newcommand{\setQ}{\ensuremath{{\mathbb Q}}}
\newcommand{\setQplus}{\ensuremath{\setQ_{\geq 0}}} %
\newcommand{\setR}{\ensuremath{\mathbb R}}
\newcommand{\setRgeqzero}{\ensuremath{\setR_{\geq 0}}}

\newcommand{\setRplus}{\ensuremath{\setR_{\geq 0}}} %
\newcommand{\setZ}{\ensuremath{\mathbb Z}}

\newcommand{\imitator}{\textsf{IMITATOR}}
\newcommand{\uppaal}{\textsc{Uppaal}}

\ifdefined \VersionWithComments
 	\definecolor{colorok}{RGB}{80,80,150}
\else
	\definecolor{colorok}{RGB}{0,0,0}
\fi

\newcommand{\eg}{\textcolor{colorok}{e.g.}\xspace}
\newcommand{\etal}{\textcolor{colorok}{\emph{et al.}}\xspace}
\newcommand{\ie}{\textcolor{colorok}{i.e.}\xspace}
\newcommand{\st}{\textcolor{colorok}{s.t.}\xspace}

\newcommand{\wrt}{\textcolor{colorok}{w.r.t.}\xspace}

\makeatletter
\def\orcidID#1{\smash{\href{https://orcid.org/#1}{\protect\raisebox{-1.25pt}{\protect\includegraphics{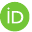}}}}}
\makeatother

\title{Efficient Convex Zone Merging\\in Parametric Timed Automata\thanks{%
	\AuthorVersion{%
		This is the author version of the manuscript of the same name published in the proceedings of the 20th International Conference on Formal Modeling and Analysis of Timed Systems (\href{https://conferences.ncl.ac.uk/formats2022/}{FORMATS 2022}).
		The final version is available at \href{https://www.doi.org/10.1007/978-3-031-15839-1_12}{ \nolinkurl{10.1007/978-3-031-15839-1_12}}.
	}
	This work is partially supported by the ANR-NRF French-Singaporean research program \href{https://www.loria.science/ProMiS/}{ProMiS} (ANR-19-CE25-0015)
	and CNRS-INS2I project TrAVAIL.
	}
}
\author{%
	\'Etienne Andr\'e\inst{1}%
	\raisebox{1ex}{\scalebox{0.8}{\orcidID{0000-0001-8473-9555}}}
	\and
	Dylan Marinho\inst{1}%
	\raisebox{1ex}{\scalebox{0.8}{\orcidID{0000-0002-2548-6196}}}
	\and
	Laure Petrucci\inst{2}%
	\raisebox{1ex}{\scalebox{0.8}{\orcidID{0000-0003-3154-5268}}}
	\and
	Jaco van de Pol\inst{3}%
	\raisebox{1ex}{\scalebox{0.8}{\orcidID{0000-0003-4305-0625}}}
}
\institute{%
	Université de Lorraine, CNRS, Inria, LORIA, F-54000 Nancy, France
	\and
	LIPN, CNRS UMR 7030, Université Sorbonne Paris Nord, Villetaneuse, France
	\and
	Aarhus University, Aarhus, Denmark
}

\begin{document}
\sloppy

\AuthorVersion{%
	\pagestyle{plain}
}

\maketitle{}

\AuthorVersion{%
	\thispagestyle{plain}
}

\begin{abstract}
	Parametric timed automata are a powerful formalism for reasoning on concurrent real-time systems with unknown or uncertain timing constants. 
	Reducing their state space is a significant way to reduce the inherently large analysis times.
	We present here different merging reduction techniques based on convex union of constraints (parametric zones), allowing to decrease the number of states while preserving the correctness of verification and synthesis results.
	We perform extensive experiments, and identify the best heuristics in practice, bringing a significant decrease in the computation time on a benchmarks library.
	
	\AuthorVersion{%
	}
\end{abstract}
\section{Introduction}\label{section:introduction}

Parametric timed automata (PTAs)~\cite{AHV93} are a powerful extension of timed automata (TAs)~\cite{AD94} with timing parameters, allowing to reason on concurrent real-time systems with unknown or uncertain timing constants.
PTAs go beyond the expressiveness of the classical model checking problem of TAs (with a binary ``yes''/``no'' answer), and can address \emph{parameter synthesis}, \ie{} the exhibition of valuations for these timing parameters \st{} a given property holds.
A common problem (addressed here) is that of \emph{reachability synthesis}: ``synthesize parameter valuations such that a given location is reachable''.

PTAs are an inherently expressive but hard formalism, in the sense that most decision problems are undecidable (see \eg{} \cite{Andre19STTT} for a survey), while verification and parameter synthesis are subject to the infamous state space explosion in practice.
Reducing the state space, built on-the-fly when performing parameter synthesis, is a significant way to reduce the sometimes large computation times.

The symbolic semantics of TAs is often represented as \emph{zones}, \ie{} linear constraints over the clocks with a special form.
In~\cite{David05}, a convex zone merging technique is presented for \uppaal{}, that preserves reachability properties.
This merging technique was extended to PTAs in~\cite{AFS13atva}, and applied to the symbolic semantics of PTAs in the form of parametric zones, \ie{} linear constraints over the clocks and the parameters, obeying to a special form~\cite{HRSV02}.
In~\cite{AFS13atva}, the analysis is only performed in the framework of the ``inverse
method'' (IM, also called ``trace preservation synthesis''~\cite{ALM20}); no other
properties are considered.

\paragraph{Contributions.}
We propose here different merging techniques for PTAs, with the goal to reduce the state space size and/or the analysis time. %
We implement our techniques in \imitator{}~\cite{Andre21}, and we perform extensive experiments on a standard benchmarks set~\cite{AMP21}.
It turns out that these various heuristics have very different outcomes in terms of size of the state space and analysis speed.
We then identify the best heuristics in practice, allowing to significantly decrease the number of states and the computation time, while preserving the correctness of the parameter synthesis for the whole class of reachability properties.
The two main differences with~\cite{AFS13atva} are
\begin{ienumerate}%
	\item the definition of merging for reachability synthesis (and not only for IM),
	and
	\item the systematic investigation of new heuristics, leading to a largely increased
	efficiency \wrt{} the original merging of~\cite{AFS13atva}.
\end{ienumerate}%

\paragraph{Related work.}
As said above, merging was first proposed for TAs in~\cite{David05}, and then extended to the ``inverse method'' for PTAs in~\cite{AFS13atva}.
In~\cite{BBM06}, Ben Salah \etal{} show that it is safe to perform the convex merging of various constraints, when they are the result of an interleaving.
The exploration is done in a BFS (breadth-first search) manner, and states are merged at each depth level.

Beyond merging, various heuristics were proposed to efficiently reduce the state space of TAs.
Extrapolation and abstractions were proposed in \cite{AD94,BBLP06,HSW13,HSW16} for TAs, and then extended to PTAs in~\cite{ALR15,BBBC16}.
Exploration orders were discussed in~\cite{HT15} and then in~\cite{ANP17} for PTAs.
The efficiency of model checking liveness properties for TAs is discussed notably in~\cite{HSW12,HSTW20}.
Zone inclusion (subsumption) for liveness checking is discussed for TAs in \cite{LODLV13} and
for PTAs in \cite{AAPP21}. Inclusion/subsumption is a special case of merging.
The other mentioned techniques are orthogonal to the merging technique, and they can be combined.

In addition, computing efficiently exact or over-approximated successors of ``zones'' %
in the larger class of \emph{hybrid automata} (HAs)~\cite{Henzinger96} is an active field of research (\eg{} \cite{CAF11,CSA14,SNA17,BFFPS20}).
Beyond the target formalism (PTAs instead of HAs), a main difference is that we are concerned here exclusively with an \emph{exact} analysis.

\AuthorVersion{%
\paragraph{Outline}
We recall the necessary concepts %
in \cref{section:preliminaries}.
Several merging heuristics are proposed in \cref{section:merging} and evaluated
in \cref{section:experiments}.
We conclude in \cref{section:conclusion}.
}

\section{Preliminaries}\label{section:preliminaries}

We assume a set~$\Clock = \{ \clock_1, \dots, \clock_\ClockCard \} $ of \emph{clocks}, \ie{} real-valued variables that evolve over time at the same rate.
A \emph{clock valuation} is a function
$\clockval : \Clock \rightarrow \setRgeqzero$.
The clock valuation $\ClocksZero$ assigns $0$ to all clocks.
Given a delay $d \in \setRgeqzero$, $\clockval + d$ denotes the valuation $(\clockval + d)(\clock) = \clockval(\clock) + d$, for $\clock \in \Clock$.
Given $\resets \subseteq \Clock$, we define the \emph{reset} of valuation~$\clockval$ by $\reset{\clockval}{\resets}(\clock) = 0$ if $\clock \in \resets$, and $\reset{\clockval}{\resets}(\clock)=\clockval(\clock)$, otherwise.

We assume a set~$\Param = \{ \param_1, \dots, \param_\ParamCard \} $ of \emph{parameters}, \ie{} unknown constants.
A linear term is of the form $\sum_{1 \leq i \leq \ClockCard} \alpha_i \clock_i + \sum_{1 \leq j \leq \ParamCard} \beta_j \param_j + d$, with
	$\alpha_i, \beta_j, d \in \setZ$.
A \emph{constraint}~$\Constr$ (\ie{} a convex polyhedron) over $\Clock \cup \Param$ is a conjunction of inequalities of the form $\lterm \compOp 0$, where $\lterm$ is a linear term
and ${\compOp} \in \{<, \leq, =, \geq, >\}$.

A \emph{parameter valuation} $\pval$ is a function $\pval : \Param \rightarrow \setQplus$. 
Given a parameter valuation~$\pval$, $\valuate{\Constr}{\pval}$ denotes the constraint over~$\Clock$ obtained by replacing each parameter~$\param$ in~$\Constr$ with~$\pval(\param)$.
Likewise, given a clock valuation~$\clockval$, $\valuate{\valuate{\Constr}{\pval}}{\clockval}$ denotes the expression obtained by replacing each clock~$\clock$ in~$\valuate{\Constr}{\pval}$ with~$\clockval(\clock)$.
We write $\clockval \models \valuate{\Constr}{\pval}$ if $\valuate{\valuate{\Constr}{\pval}}{\clockval}$ evaluates to true.
We say that $\Constr$ is \emph{satisfiable} if $\exists \clockval, \pval \text{ s.t.\ } \clockval \models \valuate{\Constr}{\pval}$.

\begin{definition}[PTA~\cite{AHV93}]\label{def:uPTA}
	A PTA $\A$ is a tuple \mbox{$\A = (\Actions, \Loc, \locinit, \Clock, \Param, \invariant, \Edges)$}, where: %
	\begin{ienumerate}
		\item $\Actions$ is a finite set of actions,
		\item $\Loc$ is a finite set of locations,
		\item $\locinit \in \Loc$ is the initial location,
		\item $\Clock$ is a finite set of clocks,
		\item $\Param$ is a finite set of parameters,
		\item $\invariant$ is the invariant, assigning to every $\loc\in \Loc$ a constraint $\invariant(\loc)$,
		\item $\Edges$ is a finite set of edges  $\edge = (\loc,\guard,\action,\resets,\loc')$
		where~$\loc,\loc'\in \Loc$ are the source and target locations, $\action \in \Actions$, $\resets\subseteq \Clock$ are the set of clocks to be reset, and the guard $\guard$ is a constraint.
	\end{ienumerate}
\end{definition}

Given\ a parameter valuation~$\pval$, $\valuate{\A}{\pval}$ denotes the non-parametric TA~\cite{AD94}, where all occurrences of any parameter~$\param_i$ have been replaced by~$\pval(\param_i)$.
\begin{definition}[Concrete semantics]\label{def:semantics:TA}
	Given PTA $\A = (\Actions, \Loc, \locinit, \Clock, \Param, \invariant, \Edges)$, %
	and a parameter valuation~\(\pval\),
	the concrete semantics of $\valuate{\A}{\pval}$ is given by the timed transition system (TTS)~\cite{HMP91,AD94} $\TTS_{\valuate{\A}{\pval}}=(\States, \sinit, \smash{{\longueflecheRel{\edge}}\cup{\longueflecheRel{d}}})$, with
	\begin{itemize}
		\item $\States = \{ (\loc, \clockval) \in \Loc \times \setRgeqzero^\ClockCard \mid \clockval \models \valuate{\invariant(\loc)}{\pval} \}$, %
		\item initial state $\sinit = (\locinit, \ClocksZero) $,
			\item discrete transitions: $(\loc,\clockval) \longueflecheRel{\edge} (\loc',\clockval')$, %
				if $(\loc, \clockval) , (\loc',\clockval') \in \States$, and there exists ${\edge = (\loc,\guard,\action,\resets,\loc') \in \Edges}$, such that $\clockval'= \reset{\clockval}{\resets}$, and $\clockval\models\pval(\guard$).
			\item delay transitions: $(\loc,\clockval) \longueflecheRel{d} (\loc, \clockval+d)$ for $d \in \setRgeqzero$, if $\forall d' \in [0, d], (\loc, \clockval+d') \in \States$.
	\end{itemize}
We write $(\loc, \clockval)\longuefleche{(d, \edge)} (\loc',\clockval')$ for a combined step:
		$\exists  \clockval'' :  (\loc,\clockval) \longueflecheRel{d} (\loc,\clockval'') \longueflecheRel{\edge} (\loc',\clockval')$.
\end{definition}

Given a TA~$\valuate{\A}{\pval}$ with concrete semantics $\TTS_{\valuate{\A}{\pval}} = (\States, \sinit, \flecheRel)$, we refer to the states of~$\TTS_{\valuate{\A}{\pval}}$ as the \emph{concrete states} of~$\valuate{\A}{\pval}$.
A \emph{run} of~$\valuate{\A}{\pval}$ is a (finite or infinite) alternating sequence of concrete states of $\valuate{\A}{\pval}$ and pairs of 
delay and discrete transitions starting from the initial state~$\sinit$ of the form
$\concstate_0, (d_0, \edge_0), \concstate_1, \cdots$
with
$i = 0, 1, \dots$, $\edge_i \in \Edges$, $d_i \in \setRgeqzero$ and
	$\smash{\concstate_i \longuefleche{(d_i, \edge_i)} \concstate_{i+1}}$.

Given a state~$\concstate = (\loc, \clockval)$, we say that $\concstate$ is reachable in~$\valuate{\A}{\pval}$ if $\concstate$ appears in a run of $\valuate{\A}{\pval}$.
By extension, we say that $\loc$ is reachable in~$\valuate{\A}{\pval}$.

\subsubsection{Symbolic Semantics.}\label{ss:symbolic}
We now recall the symbolic semantics of PTAs (see \eg{} \cite{HRSV02,ACEF09,JLR15}).
Define the \emph{time elapsing} of~$\Constr$, denoted by $\timelapse{\Constr}$, as the constraint over~$\Clock$ and~$\Param$ obtained by delaying all clocks in~$\Constr$ by an arbitrary amount of time.
That is,
\(\clockval' \models \valuate{\timelapse{\Constr}}{\pval} \text{ if } \exists \clockval : \Clock \to \setRplus, \exists d \in \setRplus \text { s.t. } \clockval \models \valuate{\Constr}{\pval} \land \clockval' = \clockval + d \text{.}\)
Given $\resets \subseteq \Clock$, define the \emph{reset} of~$\Constr$, denoted by $\reset{\Constr}{\resets}$, as the constraint obtained from~$\Constr$ by resetting the clocks in~$\resets$ to $0$, keeping other clocks unchanged.
That is,
\[\clockval' \models \valuate{\reset{\Constr}{\resets}}{\pval} \text{ if } \exists \clockval : \Clock \to \setRplus \text { s.t. } \clockval \models \valuate{\Constr}{\pval} \land \forall \clock \in \Clock
	\left \{ \begin{array}{ll}
		 \clockval'(\clock) = 0 & \text{if } \clock \in \resets\\
		 \clockval'(\clock) = \clockval(\clock) & \text{otherwise.}
	\end{array} \right .\]

We denote by $\projectP{\Constr}$ the projection of~$\Constr$ onto~$\Param$, \ie{} obtained by eliminating the variables not in~$\Param$ (\eg{} using Fourier-Motzkin).
The application of one of these operations (time elapsing, reset, projection) to a constraint yields a constraint; existential quantification can be handled, \eg{} by adding variables and subsequently eliminating them using, \eg{} Fourier-Motzkin.

	A symbolic state is a pair $(\loc, \Constr)$ where $\loc \in \Loc$ is a location, and $\Constr$ its associated constraint over~$\Clock \cup \Param$ called \emph{parametric zone}.
\begin{definition}[Symbolic semantics]\label{def:PTA:symbolic}
	Given a PTA $\A = (\Actions, \Loc, \locinit,
	\Clock, \Param, \invariant, \Edges)$, %
	the symbolic semantics of~$\A$ is the labelled transition system called \emph{parametric zone graph}
	$ \PZG = ( \Edges, \SymbStates, \symbstateinit, {\SymbTransitions} )$, with
	\begin{itemize}
		\item $\SymbStates = \{ (\loc, \Constr) \mid \Constr \subseteq \invariant(\loc) \}$, %
		\item $\symbstateinit = \big(\locinit, \timelapse{(\bigwedge_{1 \leq i\leq\ClockCard}\clock_i=0)} \land \invariant(\loc_0) \big)$,
				and
		\item $\big((\loc, \Constr), \edge, (\loc', \Constr')\big) \in {\SymbTransitions} $ if $\edge = (\loc,\guard,\action,\resets,\loc') \in \Edges$ and
			\(\Constr' = \timelapse{\big(\reset{(\Constr \land \guard)}{\resets}\land \invariant(\loc')\big )} \land \invariant(\loc')\)
			with $\Constr'$ satisfiable.
	\end{itemize}
\end{definition}

That is, in the parametric zone graph, nodes are symbolic states, and arcs are labeled by \emph{edges} of the original PTA.
Given a symbolic state~$\symbstate$ reachable in~$\PZG$, we define $\SuccE(\symbstate)$, successors with edges, by $\{ (\edge, \symbstate') \mid (\symbstate, \edge, \symbstate')\in {\SymbTransitions}\}$. 
We also write $\symbstate \SymbTransitions \symbstate'$ to denote that for some $\edge$,
$(\symbstate, \edge, \symbstate') \in {\SymbTransitions}$.
Given $\symbtrans = (\symbstate, \edge, \symbstate') \in {\SymbTransitions}$, $\symbtrans.\fieldSource$ denotes~$\symbstate$ while $\symbtrans.\fieldTarget$ denotes~$\symbstate'$.
Given $\symbstate = (\loc, \Constr)$, $\symbstate.\fieldConstr$ denotes~$\Constr$ while $\symbstate.\fieldLoc$ denotes~$\loc$.
Note that we usually use bold font to denote anything symbolic, \ie{} (sets of) symbolic states, and constraints.

A well-known result~\cite{HRSV02} is that, given a PTA~$\A$ and a reachable symbolic state $(\loc, \Constr)$, if a parameter valuation~$\pval$ belongs to the projection onto the parameters of~$\Constr$ (\ie{} $\pval \in \projectP{\Constr}$), then $\loc$ is reachable in the TA $\valuate{\A}{\pval}$.

The \emph{(symbolic) state space} of a PTA is its parametric zone graph.
This structure is in general infinite, due to the intrinsic undecidability of most decision problems for PTAs.
However, for semi-algorithms for parameter synthesis (without a guarantee of termination), 
it is of utmost importance to \emph{reduce} the size of this state space, so as to perform synthesis more efficiently.

\section{Efficient State Merging in Parametric Timed Automata}\label{section:merging}
\subsection{Merging Algorithm}

We recall the notion of merging from~\cite{AFS13atva}.
Two states are \emph{mergeable} if
\begin{oneenumerate}%
	\item they share the same location, and
	\item the union of their constraints is convex.
\end{oneenumerate}%

\begin{definition}[Merging~\cite{AFS13atva}]
	Two symbolic states $\symbstate_1 = (\loc_1, \Constr_1)$, $\symbstate_2 = (\loc_2, \Constr_2)$ are \emph{mergeable}, denoted by the predicate $\ismergeable(\symbstate_1, \symbstate_2)$, if $\loc_1 = \loc_2$ and $\Constr_1 \cup \Constr_2$ is convex.

	In that case, we define their \emph{merging} %
	as $(\loc_1, \Constr_1\cup\Constr_2)$.
\end{definition}

Merging is a generalisation of inclusion abstraction (also known as subsumption).
Note that if $\symbstate_2$ includes $\symbstate_1$, \ie{} $\Constr_1\subseteq\Constr_2$,
then $\Constr_1\cup\Constr_2=\Constr_2$ is convex, so the states can be merged, and
the result will be~$\symbstate_2$.

\begin{example}
We display examples of 2-dimensional zones in \cref{fig:merge:options:zones}. %
(These box-shaped parametric zones are fictitious and displayed for the purpose of illustration; similar zones, sometimes using ``diagonal'' edges, can be obtained from actual PTAs.)
Zone~$\Constr_1$ can be merged with~$\Constr_4$; $\Constr_2$ can also be merged with~$\Constr_4$.
The result of these two merging operations is shown in \cref{fig:merge:options:m31}.
These two new zones can also be merged together, leading to the zone in \cref{fig:merge:options:m32}.
\end{example}

\begin{example} Let us now consider the PTA in \cref{fig:merge:ex2:PTA}, with
 two clocks ($x$ and $y$) and two parameters ($p$ and $q$).
Both clocks and parameters are initially bound to be non-negative (clocks initially 
different from~0 can be simulated using an appropriate gadget, omitted here).

The PZG of this PTA is shown in \cref{fig:merge:ex2:PZG}.
It features two separate infinite executions which depend on the first chosen transition.
In the upper branch, the first state with location~$\loc_1$ has constraint $p + x
\geq y$ (which can read $y - x \leq p$) since, although we have $y \leq p$ when taking
the transition from~$\loc_0$, time can then elapse in~$\loc_1$---but only up to $p$ time units, due to invariant $x \leq p$.
Then coefficients on~$p$ (\ie{} $2p$ then $3p$, etc.)\ start to appear from the second state with location~$\loc_1$ due to the self-loop on~$\loc_1$ that resets~$x$.

\emph{Inclusion} reduces one of these two symbolic executions, which exhibits
decreasing zones, as in \cref{fig:merge:ex2:PZGincl}.
Even using inclusion, the PZG remains infinite.

Finally, \cref{fig:merge:ex2:PZGmerge} displays the graph obtained with the \emph{merging} approach:
the two states obtained after taking a single transition can be merged.
Here, the PZG with merging becomes finite, which illustrates the importance of merging.

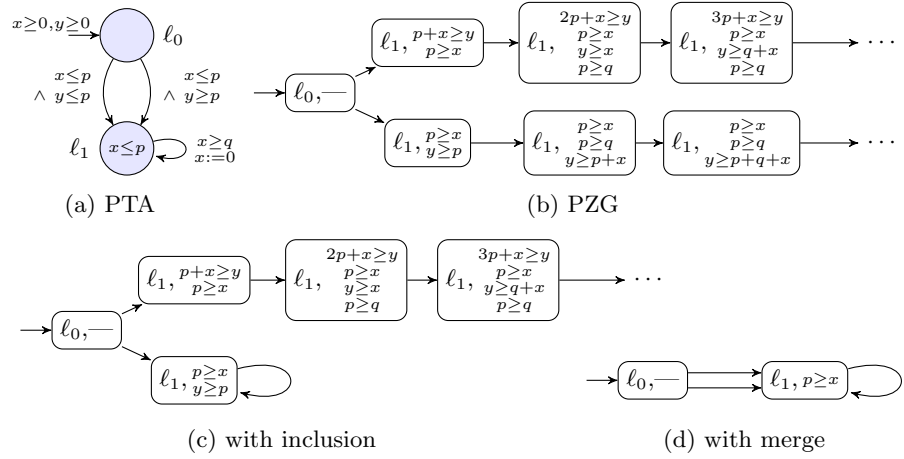
\begin{figure}[tb]
\centering
\begin{subfigure}[b]{0.23\textwidth}
	\centering
	\begin{tikzpicture}[pta, node distance =1.5cm]
	\node[location,initial,
		label={[below,xshift=-1cm]:$\scriptstyle x\geq0, y\geq0$},
		label={right:$\loc_0$}] (q0)
		{\phantom{$\scriptstyle x\leq p$}};
	\node[location,label={left:$\loc_1$}] (q1) [below of=q0]
		{$\scriptstyle x\leq p$};
	\path
		(q0) edge [bend right]
			node [left] {\begin{tabular}{ll}
										 & $\scriptstyle x\leq p$ \\ 
										\raisebox{2mm}{$\scriptstyle\land$} &
											\raisebox{2mm}{$\scriptstyle y\leq p$}
										\end{tabular}} (q1)
		(q0) edge [bend left]
			node [right] {\begin{tabular}{ll}
										 & $\scriptstyle x\leq p$ \\ 
										\raisebox{2mm}{$\scriptstyle\land$} &
											\raisebox{2mm}{$\scriptstyle y\geq p$}
										\end{tabular}} (q1)
		(q1) edge [loop right]
			node [right] {$\scriptstyle \substack{x\geq q\\ x:=0}$} (q1);
	\end{tikzpicture}
	\caption{\label{fig:merge:ex2:PTA}PTA}
\end{subfigure}
\begin{subfigure}[b]{0.75\textwidth}
	\centering
	\begin{tikzpicture}[PZG, node distance=2cm]
		\node[symbstate,initial] (q0)
			{$\loc_0,$---};
		\node[symbstate] (q1)
			[above right of=q0,yshift=-.75cm]
			{$\loc_1,\substack{p+x\geq y \\ p\geq x}$};
		\node[symbstate] (q2)
			[below right of=q0,yshift=0.75cm]
			{$\loc_1,\substack{p\geq x \\ y\geq p}$};
		\node[symbstate] (q3)
			[right of=q1]
			{$\loc_1,
				\substack{2p+x\geq y \\ p\geq x \\ y\geq x \\ p\geq q \\ }$};
		\node[symbstate] (q4)
			[right of=q2]
			{$\loc_1,\substack{p\geq x \\ p\geq q \\ y\geq p+x}$};
		\node[symbstate] (q5)
			[right of=q3]
			{$\loc_1,
				\substack{3p+x\geq y \\ p\geq x \\ y\geq q+x \\ p\geq q \\ }$};
		\node[symbstate] (q6)
			[right of=q4]
			{$\loc_1,\substack{p\geq x \\ p\geq q \\ y\geq p+q+x}$};
		\node[rectangle] (q7)
			[right of=q5]	{$\cdots$};
		\node[rectangle] (q8)
			[right of=q6]	{$\cdots$};
		\path[->]
			(q0) edge (q1) (q0) edge (q2)
			(q1) edge (q3)
			(q2) edge (q4)
			(q3) edge (q5)
			(q4) edge (q6)
			(q5) edge (q7)
			(q6) edge (q8)
		;
	\end{tikzpicture}
	\caption{\label{fig:merge:ex2:PZG}PZG}
\end{subfigure}

\medskip

\begin{subfigure}[b]{0.6\textwidth}
	\centering
	\begin{tikzpicture}[PZG, node distance=2cm]
		\node[symbstate,initial] (q0)
			{$\loc_0,$---};
		\node[symbstate] (q1)
			[above right of=q0,yshift=-0.75cm]
			{$\loc_1,\substack{p+x\geq y \\ p\geq x}$};
		\node[symbstate] (q2)
			[below right of=q0,yshift=0.75cm]
			{$\loc_1,\substack{p\geq x \\ y\geq p}$};
		\node[symbstate] (q3)
			[right of=q1]
			{$\loc_1,
				\substack{2p+x\geq y \\ p\geq x \\ y\geq x \\ p\geq q \\ }$};
		\node[symbstate] (q5)
			[right of=q3]
			{$\loc_1,
				\substack{3p+x\geq y \\ p\geq x \\ y\geq q+x \\ p\geq q \\ }$};
		\node[rectangle] (q7)
			[right of=q5]	{$\cdots$};
		\path[->]
			(q0) edge (q1) (q0) edge (q2)
			(q1) edge (q3)
			(q2) edge[loop right] (q2)
			(q3) edge (q5)
			(q5) edge (q7)
		;
	\end{tikzpicture}
	\caption{\label{fig:merge:ex2:PZGincl}with inclusion}
\end{subfigure}
\begin{subfigure}[b]{0.38\textwidth}
	\centering
	\begin{tikzpicture}[PZG, node distance=2cm]
		\node[symbstate,initial] (q0)
			{$\loc_0,$---};
		\node[symbstate] (q1)
			[right of=q0]
			{$\loc_1,\substack{p\geq x}$};
		\draw[->]	($(q0.east)+(0,0.1)$) -- ($(q1.west)+(0,0.1)$);
		\draw[->]	($(q0.east)+(0,-0.1)$) -- ($(q1.west)+(0,-0.1)$);
		\draw (q1) edge[loop right] (q1);
	\end{tikzpicture}
	\caption{\label{fig:merge:ex2:PZGmerge}with merge}
\end{subfigure}
\caption{\label{fig:merge:ex2}Example with infinite PZG that becomes finite by
 merging}
\end{figure}
\end{example}

\begin{algorithm}[tb]
 	\tcc{Building $\PZG = ( \Edges, \visited, \symbstateinit, {\SymbTransitions} )$}

		$\visited \assign \{\symbstateinit\}$
		\LongVersion{%
		
		}\ShortVersion{;}
		$\queue \assign \{\symbstateinit\}$
		\LongVersion{%
		
		}\ShortVersion{;}
		${\SymbTransitions} \assign \emptyset$
		
		\While{$\queue \neq \emptyset$}{
			$\queuenew \assign \emptyset$
			
			\ForEach{$\symbstate \in \queue$}{
				
				\ForEach{$(\edge, \symbstate') \in \SuccE(\symbstate)$}{
					$\queuenew \assign \queuenew \cup (\{\symbstate'\} \setminus \visited)$\nllabel{alg:BFS-by-layer:updateS}

					${\SymbTransitions} \assign {\SymbTransitions} \cup \{ (\symbstate, \edge, \symbstate') \}$\nllabel{alg:BFS-by-layer:updateT}
				}
			}
		
			$\visited \assign \visited \cup \queuenew$
			
			$\visited, \queue \assign \mergeSets(\PZG, \visited, \queuenew)$\nllabel{alg:BFS-by-layer:merge}
		}
	
	\caption{BFS by layer $\bfsbylayer(\A)$}
	\label{alg:BFS-by-layer}
\end{algorithm}

\cref{alg:BFS-by-layer} constructs the state space for a given PTA~$\A$ by
breadth-first search (BFS) from the initial state~$\symbstateinit$.
It computes the set of reachable states \visited{} by repeatedly adding the next layer of successor states~$\queuenew$ (\cref{alg:BFS-by-layer:updateS}), maintaining the transitions (\cref{alg:BFS-by-layer:updateT}).
Note that each iteration (\cref{alg:BFS-by-layer:merge}) calls a merging function $\mergeSets$ (given in \cref{alg:merge_function}), which may reduce both \visited{} and \queue{}.
This call to the merging function is the crux of our approach.

\cref{alg:BFS-by-layer} can be extended, depending on the analysis or parameter synthesis problem.
For instance, an invariant property can be checked for each reachable
symbolic
state and terminate as soon as the property is violated.
For reachability synthesis, one may accumulate all solutions as a set of constraints that lead to a state satisfying a property (this algorithm, $\EFsynth$, was formalized in \eg{} \cite{JLR15}).
\LongVersion{%
	As an optimization, one can prune the state space on reaching states that lie already within the accumulated constraints,  since these cannot lead to new solutions.
	We do not add such extensions to the algorithm, since this paper focuses on the merging of states during state space generation.
}

Then, \cref{alg:merge_function} ``simply'' calls recursively the $\mergeonestate$ function (given in \cref{algo:mergeOneState}) on each state of $\queue$, using additional arguments $\PZG$ and $\visited$ and/or $\queue$.
The heuristics to select arguments for calls to $\mergeonestate$ will be discussed later.
Note that $\mergeonestate$ \emph{modifies} its arguments, notably~$\PZG$
(in the implementation, we use a call by reference).

\begin{algorithm}[tb]

	\Fn{$\mergeSets(\PZG, V, Q)$}{
	\ForEach{$\symbstate\in Q$}{
	\setlength{\fboxrule}{1pt}
	\fcolorbox{blue}{white}{
		\begin{varwidth}{5cm}
			\HiLi[yellow!40]{$\mergeonestate(\symbstate, \PZG, Q)$}
			
			\HiLi[red!20]{$\mergeonestate(\symbstate, \PZG, V)$}
		\end{varwidth}
	}

	}

	\Return $(V, Q)$
	
	}
	\caption{Heuristics to merge states within $Q$ and/or $V$ \\
		\HiLi[red!20]{Visited\strut}
		\HiLi[yellow!40]{Queue\strut}
		\setlength{\fboxrule}{1pt}
		\fcolorbox{blue}{white}{Ordered}
	}
	\label{alg:merge_function}
\end{algorithm}

\begin{algorithm}[tb]
	\Fn{$\mergeonestate(\symbstate, \PZG, \looksiblings)$}{
		
		\HiLi[blue!20]{$\ismerged \assign \BFalse$ \strut}

		$\candidates \assign \getsiblings(\symbstate, \looksiblings)$\nllabel{algo:mergeOneState:siblings}
		
		\ForEach{$\symbstatey \in \candidates$}{
			\tcc{Mergeability test}
			\If{$\ismergeable(\symbstate , \symbstatey)$\nllabel{algo:mergeOneState:mergeable}}{
				$\symbstate.\fieldConstr \assign \symbstate.\fieldConstr \cup \symbstatey.\fieldConstr$\nllabel{algo:mergeOneState:merge} %
				
				\HiLi[blue!20]{$\ismerged \assign \BTrue$ \strut}
				
				\tcc{Update transition targets and source}
				
				\ForEach{$\symbtrans \in {\SymbTransitions}$\nllabel{algo:mergeOneState:deleteT:begin}}{
					\lIf{$\symbtrans.\fieldTarget = \symbstatey$}{
						$\symbtrans.\fieldTarget \assign \symbstate$
					}
					\lIf{$\symbtrans.\fieldSource = \symbstatey$}{
						$\symbtrans.\fieldSource \assign \symbstate$\nllabel{algo:mergeOneState:deleteT:end}
					}
				}

				\tcc{Delete $\symbstatey$}
				$\SymbStates \assign \SymbStates \setminus \{ \symbstatey \}$\nllabel{algo:mergeOneState:deleteS}
				
				\tcc{Handle initial state}
				\lIf{$\symbstateinit = \symbstatey$}{
					$\symbstateinit \assign \symbstate$\nllabel{algo:mergeOneState:init}
				}
			}
		}
		
		\HiLi[blue!20]{
			\lIf{$\ismerged$}{
				$\mergeonestate(\symbstate, \PZG, \looksiblings)$
		}	
		\strut}

	}
	
	\caption{Merging a state with an update of the statespace $\PZG = ( \Edges, \SymbStates, \symbstateinit, {\SymbTransitions} )$ on-the-fly. The variant with restart after a merge is indicated as in
		\HiLi[blue!20]{Restart\strut}}
	\label{algo:mergeOneState}
\end{algorithm}

\cref{algo:mergeOneState} attempts at merging a state~$\symbstate$ while looking for candidate states in \looksiblings{}.
We first look for the \emph{siblings} of~$\symbstate$ (states with same location) within \looksiblings{} (\cref{algo:mergeOneState:siblings}).
We use a function $\getsiblings((\loc, \Constr), \SymbStates)$ that returns the siblings, as in
$\{ (\loc', \Constr') \in \SymbStates \mid \loc = \loc' \}$.
If the union of the constraints is convex (\cref{algo:mergeOneState:mergeable}), state~$\symbstate$ becomes the result of  merging~$\symbstate$ with the candidate~$\symbstatey$ (\cref{algo:mergeOneState:merge}).
The candidate~$\symbstatey$ is deleted (\cref{algo:mergeOneState:deleteS}),
as well as all transitions leading to or coming from it (\crefrange{algo:mergeOneState:deleteT:begin}{algo:mergeOneState:deleteT:end}).
We finally modify the initial state~$\symbstateinit$ of~$\PZG$, in case it was merged (\cref{algo:mergeOneState:init}).

\subsection{Heuristics for Merging}
\label{subsec:options}

We now introduce and discuss several heuristics for merging states, leading to various options in the merging algorithm.
There is no provably best option that is guaranteed to be superior over all other possible options.
We will perform extensive experiments in \cref{section:experiments} to find out what works well on a number of benchmarks.
The two main driving forces to select between these options are:
\begin{oneenumerate}%
	\item a maximal reduction of the state space; and
	\item a minimisation of the computation time.
\end{oneenumerate}%
Although usually smaller state spaces tend to require less computation time, this
is not always the case: Sometimes one might need extra effort to check if states can
be merged, in order to perform even more reduction.
Subsequent computations might
profit from the smaller state space, but if one is checking properties on-the-fly, the
extra effort might not be justifiable. The discussion on the options will be guided
by some questions. 
A question that we have not investigated is if it is advantageous to merge triples of states.

The subsequent \cref{ex:merge:options} will show that different choices can indeed lead to state spaces of different size.
Note that, even when we fix the answers to the questions, the result is still non-deterministic, since the result of merging depends on the \emph{order} in which we would consider the siblings.

\paragraph{Question 1: What to merge with what?} Assume that we are computing
the next level of reachable states in a BFS process (\cref{alg:BFS-by-layer}).
Assuming that the states in \visited have been properly merged, we clearly still need
to merge the new states in $\symbstate\in\queue$. What to merge them with? Do we only
compare $\symbstate$ with other states in the \queue? Or also with
\visited? If we merge with \visited states, the final state space could become
smaller. On the other hand, since time was spent to compute those states already,
is it worth looking at them?
The different strategies considered merge a new state with its siblings:

\begin{itemize}
	\item only in the queue (\namequeue{}), or
	\item in all visited states (\namevisited{}) (including the queue), or
	\item first in the queue and, after that, in the visited list (\nameordered{}).
\end{itemize}

\noindent These different possibilities are pictured by different colours in
\cref{alg:merge_function}.

\paragraph{Question 2: Restart after a merge?} The next question is what to do if we find that $\symbstate$ could be merged
with some $\symbstate'$ into the (larger) $\symbstate_m$? We have already searched through some set $Q'$ of states before we found $\symbstate'$. Those states in $Q'$ could not be merged with $\symbstate$. However, it could be possible that a state in $\symbstate''\in Q'$ can be merged with $\symbstate_m$.
So should we {\em restart} the search (and lose some time to find more reduction),
or should we just resume the search, and only find merge candidates for $\symbstate_m$
in the remaining states that we have not yet considered? So, if a state can be merged with one of its siblings, 
should we restart or not restart the search through all candidate siblings? %

\paragraph{Question 3: When to update the statespace?} 
Assume that we find a successful merge of a state $\symbstate\in\queue$
with some other state $\symbstate'\in\visited$, leading to a larger state~$\symbstate_m$.
How do we now modify the 
already computed part of the state space?
We replace $\symbstate'$ by~$\symbstate_m$, redirecting all transitions going to~$\symbstate'$
to $\symbstate_m$. This could make the successors of~$\symbstate'$ unreachable, so we could also redirect
transitions from~$\symbstate'$ to transitions from~$\symbstate_m$. Alternatively, we could just remove the
successors of~$\symbstate'$. This is valid, since we will still compute all successors of~$\symbstate_m$
in the next level. Similar considerations apply to all states reachable from successors of~$\symbstate'$.

These approaches can have unforeseen effects: First of all, if we remove successor states, they cannot act anymore as merge candidates, thus potentially blocking future merges.
Second, removing transitions may change the ``shortest path'' to reachable states, leading to wrong answers for depth-bounded and shortest-path searches. Third, not removing states leads to a larger state space than necessary. Finally, doing a full reachability analysis is linear in the size of the state space generated so far (but does not involve any polyhedra computations). 

So one question is how often we should update the computed part of the state space?
The options we considered are to do ``garbage collection'' after: %
\begin{itemize}
	\item each merge with each sibling, or
	\item having processed the whole candidate list of a state, or
	\item having processed all states in a complete level.
\end{itemize}

\paragraph{Question 4: How to update the statespace?} The ``garbage collection'' can be implemented in two ways: If we can merge a state, we update the statespace:
\begin{itemize}
	\item reconstruct: with a copy of the reachable part of the statespace, or
	\item on-the-fly: deleting the merged state and updating its transitions \emph{in situ}.
\end{itemize}
Deleting states on-the-fly is cheaper than running a separate algorithm to mark and copy the reachable
part of the state space.
However, note that when updating transitions on-the-fly, some unnecessary successor states may stay in the state space. 
These unnecessary states and transitions lead to a waste of memory. 
On the other hand, they might still be useful as merge candidates for future merges.

\begin{example}\label{ex:merge:options}%
Recall that \cref{fig:merge:options} presents a fictitious example summarising the effect of the
options discussed in this section.
The parametric zone graph (with five states) is shown in \cref{fig:merge:options:pzg}, the corresponding projections of the zones on the parameters in \cref{fig:merge:options:zones} and the legend for the different colours in \cref{fig:merge:options:legend}.
All states have the same location, hence may be candidates for merging.
Two states ($\symbstate_0$ and $\symbstate_1$) are in the visited set, while two are in the queue ($\symbstate_2$
and $\symbstate_3$), and the last one, $\symbstate_4$, is currently being handled.

Let us first consider that the merge is only done with states in the queue. Then $\symbstate_4$
is merged with $\symbstate_2$ and no merge with $\symbstate_3$ can occur. This leads to the zones depicted
in \cref{fig:merge:options:m1}.

Let us now consider that the merge is done with all visited states. 
Then the following execution becomes possible:
State $\symbstate_4$ could be first merged with $\symbstate_1$, leading to the zones in \cref{fig:merge:options:m21}.
Now, if the restart option is used, this newly computed zone could be merged with
$\symbstate_0$, leading to the zones in \cref{fig:merge:options:m22}.
Note that we cannot merge the result with $\symbstate_2$ anymore.

Finally, let us consider the case where we merge with the queue first and then with
visited. State $\symbstate_4$ is then merged with $\symbstate_2$, as in \cref{fig:merge:options:m1}.
Then no merge with $\symbstate_3$ nor with $\symbstate_0$ can be performed, but a merge with $\symbstate_1$ is
possible, leading to \cref{fig:merge:options:m31}. If furthermore the state space
is updated immediately after a merge, the new state (instead of $\symbstate_4$) is merged with
$\symbstate_1$, leading to \cref{fig:merge:options:m32}.

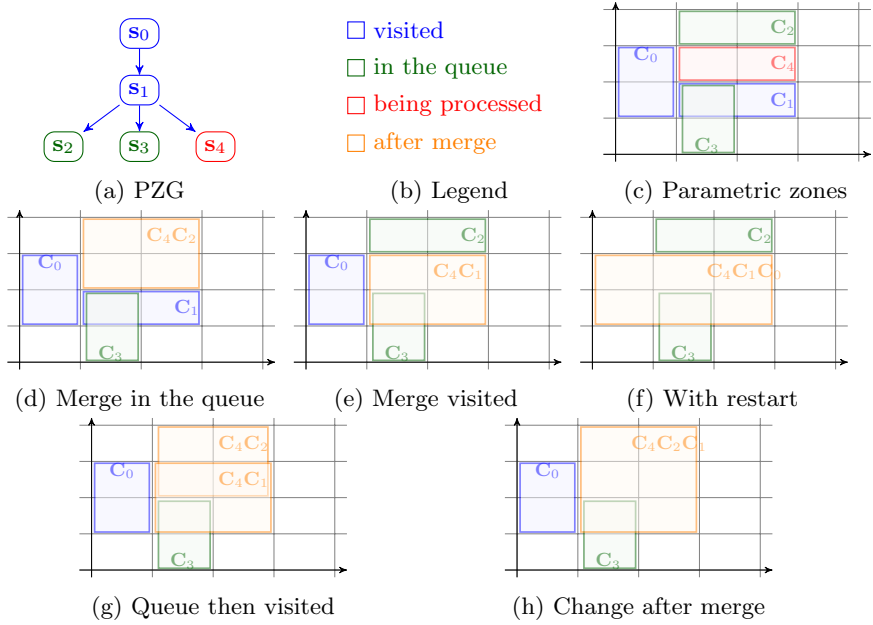
\begin{figure}[tb]
\centering
\begin{subfigure}[b]{0.35\textwidth}
	\centering
	\begin{tikzpicture}[PZG, yscale=.5]
	\begin{scope}[blue]
		\node[symbstate]
			(s0){$\symbstate_0$};
		\node at (0,-1.5) [symbstate]
			(s1){$\symbstate_1$};
		\draw (s0) -- (s1);
	\end{scope}
	\begin{scope}[DarkGreen]
		\node at (-1,-3) [symbstate]
			(s2){$\symbstate_2$};
		\node at (0,-3) [symbstate]
			(s3){$\symbstate_3$};
	\end{scope}
	\begin{scope}[red]
		\node at (1,-3) [symbstate]
			(s4){$\symbstate_4$};
	\end{scope}
	\begin{scope}[blue]
		\draw (s1) -- (s2);
		\draw (s1) -- (s3);
		\draw (s1) -- (s4);
	\end{scope}
	\end{tikzpicture}
	\caption{PZG\label{fig:merge:options:pzg}}
\end{subfigure}
\begin{subfigure}[b]{.3\textwidth}
	\centering
	\begin{tikzpicture}[PZG, node distance=2cm]
	\begin{scope}[blue]
		\node[rectangle,draw,label=right:visited](lblue){};
	\end{scope}
	\begin{scope}[DarkGreen]
		\node at (0, -.5) [rectangle,draw,label=right:{in the queue}](lgreen){};
	\end{scope}
	\begin{scope}[red]
		\node at (0, -1) [rectangle,draw,label=right:being processed](lred){};
	\end{scope}
	\begin{scope}[orange]
		\node at (0, -1.5) [rectangle,draw,label=right:after merge](lorange){};
	\end{scope}
	\end{tikzpicture}
	\caption{Legend\label{fig:merge:options:legend}}
\end{subfigure}
\begin{subfigure}[b]{.3\textwidth}
\centering
	\scalebox{0.8}{
	\begin{tikzpicture}[mergingFigure]
  \draw[gray,very thin] (-0.2,-0.2) grid (4.2,4.2);
  \draw[->] (-0.2,0) -- (4.2,0);
  \draw[->] (0,-0.2) -- (0,4.2);
	\begin{scope}[blue,fill=blue!5,opacity=0.5]
		\filldraw[thick] (0.05,1.05) rectangle (0.95,2.95);
		\node at (0.5,2.75) (z0) {$\Constr_0$};
		\filldraw[thick] (1.05,1.05) rectangle (2.95,1.95);
		\node at (2.75,1.5) (z1) {$\Constr_1$};
	\end{scope}
	\begin{scope}[DarkGreen,fill=DarkGreen!5,opacity=0.5]
		\filldraw[thick] (1.05,3.05) rectangle (2.95,3.95);
		\node at (2.75,3.5) (z2) {$\Constr_2$};
		\filldraw[thick] (1.1,0.05) rectangle (1.95,1.9);
		\node at (1.5,0.25) (z3) {$\Constr_3$};
	\end{scope}
	\begin{scope}[red,fill=red!5,opacity=0.5]
		\filldraw[thick] (1.05,2.05) rectangle (2.95,2.95);
		\node at (2.75,2.5) (z4) {$\Constr_4$};
	\end{scope}
	\end{tikzpicture}
	}
	\caption{Parametric zones\label{fig:merge:options:zones}}
\end{subfigure}

\begin{subfigure}[b]{0.3\textwidth}
\centering
	\scalebox{0.8}{
	\begin{tikzpicture}[mergingFigure]
  \draw[gray,very thin] (-0.2,-0.2) grid (4.2,4.2);
  \draw[->] (-0.2,0) -- (4.2,0);
  \draw[->] (0,-0.2) -- (0,4.2);
	\begin{scope}[blue,fill=blue!5,opacity=0.5]
		\filldraw[thick] (0.05,1.05) rectangle (0.95,2.95);
		\node at (0.5,2.75) (z0) {$\Constr_0$};
		\filldraw[thick] (1.05,1.05) rectangle (2.95,1.95);
		\node at (2.75,1.5) (z1) {$\Constr_1$};
	\end{scope}
	\begin{scope}[DarkGreen,fill=DarkGreen!5,opacity=0.5]
		\filldraw[thick] (1.1,0.05) rectangle (1.95,1.9);
		\node at (1.5,0.25) (z3) {$\Constr_3$};
	\end{scope}
	\begin{scope}[orange,fill=orange!5,opacity=0.5]
		\filldraw[thick] (1.05,2.05) rectangle (2.95,3.95);
		\node at (2.5,3.5) (z4z2) {$\Constr_4\Constr_2$};
	\end{scope}
	\end{tikzpicture}
	}
	\caption{Merge in the queue\label{fig:merge:options:m1}}
\end{subfigure}
\begin{subfigure}[b]{0.3\textwidth}
\centering
	\scalebox{0.8}{
	\begin{tikzpicture}[mergingFigure]
  \draw[gray,very thin] (-0.2,-0.2) grid (4.2,4.2);
  \draw[->] (-0.2,0) -- (4.2,0);
  \draw[->] (0,-0.2) -- (0,4.2);
	\begin{scope}[blue,fill=blue!5,opacity=0.5]
		\filldraw[thick] (0.05,1.05) rectangle (0.95,2.95);
		\node at (0.5,2.75) (z0) {$\Constr_0$};
	\end{scope}
	\begin{scope}[DarkGreen,fill=DarkGreen!5,opacity=0.5]
		\filldraw[thick] (1.05,3.05) rectangle (2.95,3.95);
		\node at (2.75,3.5) (z2) {$\Constr_2$};
		\filldraw[thick] (1.1,0.05) rectangle (1.95,1.9);
		\node at (1.5,0.25) (z3) {$\Constr_3$};
	\end{scope}
	\begin{scope}[orange,fill=orange!5,opacity=0.5]
		\filldraw[thick] (1.05,1.05) rectangle (2.95,2.95);
		\node at (2.5,2.5) (z4z1) {$\Constr_4\Constr_1$};
	\end{scope}
	\end{tikzpicture}
	}
	\caption{Merge visited\label{fig:merge:options:m21}}
\end{subfigure}
\begin{subfigure}[b]{0.3\textwidth}
\centering
	\scalebox{0.8}{
	\begin{tikzpicture}[mergingFigure]
  \draw[gray,very thin] (-0.2,-0.2) grid (4.2,4.2);
  \draw[->] (-0.2,0) -- (4.2,0);
  \draw[->] (0,-0.2) -- (0,4.2);
	\begin{scope}[DarkGreen,fill=DarkGreen!5,opacity=0.5]
		\filldraw[thick] (1.05,3.05) rectangle (2.95,3.95);
		\node at (2.75,3.5) (z2) {$\Constr_2$};
		\filldraw[thick] (1.1,0.05) rectangle (1.95,1.9);
		\node at (1.5,0.25) (z3) {$\Constr_3$};
	\end{scope}
	\begin{scope}[orange,fill=orange!5,opacity=0.5]
		\filldraw[thick] (0.05,1.05) rectangle (2.95,2.95);
		\node at (2.5,2.5) (z4z1z0) {$\Constr_4\Constr_1\Constr_0$};
	\end{scope}
	\end{tikzpicture}
	}
	\caption{With restart\label{fig:merge:options:m22}}
\end{subfigure}

\begin{subfigure}[b]{0.45\textwidth}
\centering
	\scalebox{0.8}{
	\begin{tikzpicture}[mergingFigure]
  \draw[gray,very thin] (-0.2,-0.2) grid (4.2,4.2);
  \draw[->] (-0.2,0) -- (4.2,0);
  \draw[->] (0,-0.2) -- (0,4.2);
	\begin{scope}[blue,fill=blue!5,opacity=0.5]
		\filldraw[thick] (0.05,1.05) rectangle (0.95,2.95);
		\node at (0.5,2.75) (z0) {$\Constr_0$};
	\end{scope}
	\begin{scope}[DarkGreen,fill=DarkGreen!5,opacity=0.5]
		\filldraw[thick] (1.1,0.05) rectangle (1.95,1.9);
		\node at (1.5,0.25) (z3) {$\Constr_3$};
	\end{scope}
	\begin{scope}[orange,fill=orange!5,opacity=0.5]
		\filldraw[thick] (1.1,2.05) rectangle (2.9,3.95);
		\node at (2.5,3.5) (z4z2) {$\Constr_4\Constr_2$};
		\filldraw[thick] (1.05,1.05) rectangle (2.95,2.95);
		\node at (2.5,2.5) (z4z1) {$\Constr_4\Constr_1$};
	\end{scope}
	\end{tikzpicture}
	}
	\caption{Queue then visited\label{fig:merge:options:m31}}
\end{subfigure}
\begin{subfigure}[b]{0.45\textwidth}
\centering
	\scalebox{0.8}{
	\begin{tikzpicture}[mergingFigure]
  \draw[gray,very thin] (-0.2,-0.2) grid (4.2,4.2);
  \draw[->] (-0.2,0) -- (4.2,0);
  \draw[->] (0,-0.2) -- (0,4.2);
	\begin{scope}[blue,fill=blue!5,opacity=0.5]
		\filldraw[thick] (0.05,1.05) rectangle (0.95,2.95);
		\node at (0.5,2.75) (z0) {$\Constr_0$};
	\end{scope}
	\begin{scope}[DarkGreen,fill=DarkGreen!5,opacity=0.5]
		\filldraw[thick] (1.1,0.05) rectangle (1.95,1.9);
		\node at (1.5,0.25) (z3) {$\Constr_3$};
	\end{scope}
	\begin{scope}[orange,fill=orange!5,opacity=0.5]
		\filldraw[thick] (1.05,1.05) rectangle (2.95,3.95);
		\node at (2.5,3.5) (z4z2z1) {$\Constr_4\Constr_2\Constr_1$};
	\end{scope}
	\end{tikzpicture}
	}
	\caption{Change after merge\label{fig:merge:options:m32}}
\end{subfigure}

\caption{Illustration of the merging options\label{fig:merge:options}}
\end{figure}
\end{example}

\subsection{Preservation of Properties}

\begin{prop}
Given a PTA $\A$, let $\PZG$ and $\PZG'$ be the parametric zone graph before and after merging.
Then $\PZG'$ simulates $\PZG$.
\end{prop}
\begin{proof}[sketch]
Consider the relation $\symbstate \sqsubseteq \symbstate'$ if and only if
$\symbstate.\fieldLoc = \symbstate'.\fieldLoc$ and 
$\symbstate.\fieldConstr \subseteq \symbstate'.\fieldConstr$.
It is well-known \cite{NPP18} that this forms a simulation relation, \ie{}
if $\symbstate \sqsubseteq \symbstate'$ and $\symbstate \SymbTransitions  \symbstate_0$, %
then for some $\symbstate'_0$, we have $\symbstate_0 \sqsubseteq \symbstate'_0$
and $\symbstate' \SymbTransitions  \symbstate'_0$.

Note that while merging, we repeatedly replace a state $(\loc,\Constr)$ by a state $(\loc,\Constr\cup \Constr')$,
in which case $(\loc,\Constr) \sqsubseteq (\loc,\Constr\cup \Constr')$. So indeed, merged states
can simulate the behaviour of all original states that were merged.
\end{proof}

\begin{corollary}
Given a PTA $\A$, let $\PZG$ and $\PZG'$ be the parametric zone graph before and after merging.
Let $\varphi$ be a property in $\forall$CTL* with atomic propositions defined in terms
of state locations only. Then $\PZG'\vDash \varphi$ implies $\PZG\vDash \varphi$.
\end{corollary}
\begin{proof}[sketch]
All universal properties (in $\forall$CTL*) are preserved by simulation \cite[Thm 7.76]{BKPoMC08}.
In this case, the simulation $\sqsubseteq$ also implies that related states
have the same locations, so they satisfy the same atomic properties.
\end{proof}

The next proposition shows that we do not add arbitrary new behaviour.
Although merging can add behaviour, it cannot add unreachable locations and, more precisely, the set of locations reachable for each parameter valuation remains unchanged.
This guarantees that merging preserves reachability synthesis.
(A version of this result was shown in a different context in \cite[Thm.~1]{AFS13atva}%
.)

\begin{restatable}[preservation of reachability properties]{prop}{propositionSynthesis}\label{proposition:synthesis}
	Given a PTA~$\A$, let $\PZG'$ be the parametric zone graph after merging.
	Let~$\loc$ be a location, let~$\pval$ be a parameter valuation.
	
	$\loc$ is reachable in $\valuate{\A}{\pval}$ iff $\exists (\loc, \Constr') \in \PZG'$ such that $\pval \in \projectP{\Constr'}$.
\end{restatable}

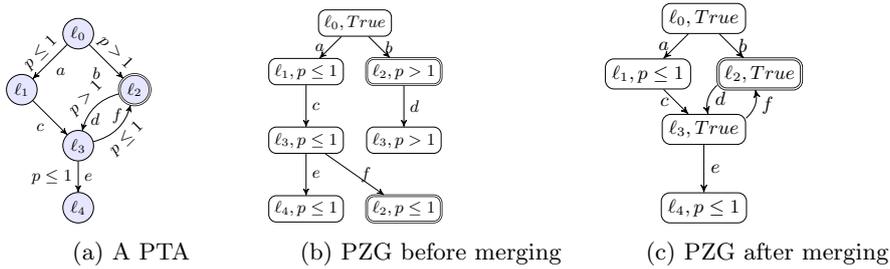
\begin{figure}[b]
\hfill
\begin{subfigure}[b]{0.27\textwidth}
\scalebox{.7}{\begin{tikzpicture}[pta, node distance=1.5cm]
\node[location] (l0) {$\loc_0$};
\node[location, below left of=l0] (l1) {$\loc_1$};
\node[location, accepting, below right of=l0] (l2) {$\loc_2$};
\node[location, below right of=l1] (l3) {$\loc_3$};
\node[location, below of=l3, yshift=1em] (l4) {$\loc_4$};
\path[]
  (l0) 
  	edge node [above,sloped] {$p\leq 1$} node [below right] {$a$} (l1)
  	edge node [above,sloped] {$p > 1$} node [below left] {$b$} (l2)
  (l1) edge node [below left] {$c$}  (l3)
  (l2) edge [bend right] node [above,sloped] {$p > 1$} node[below] {$d$} (l3)
  (l3) edge [bend right] node [below,sloped] {$p\leq 1$} node[above] {$f$} (l2)
  (l3) edge node [left] {$p\leq 1$} node [right] {$e$} (l4)
  ;
\end{tikzpicture}}
\caption{A PTA}
\label{fig:PTAex}
\end{subfigure}
\hfill
\begin{subfigure}[b]{0.35\textwidth}
\scalebox{0.7}{\begin{tikzpicture}[PZG, node distance=1.3cm]
\node[symbstate] (l0) {$\loc_0,True$};
\node[symbstate, below left of=l0] (l1) {$\loc_1,p\leq 1$};
\node[symbstate, accepting, below right of=l0] (l2a) {$\loc_2, p> 1$};
\node[symbstate, below of= l1] (l3a) {$\loc_3, p\leq 1$};
\node[symbstate, below of= l2a] (l3b) {$\loc_3, p> 1$};
\node[symbstate, below of= l3a] (l4) {$\loc_4,p\leq 1$};
\node[symbstate, accepting, below of= l3b] (l2b) {$\loc_2, p\leq 1$};

\path[]
  (l0) edge node [left] {$a$} (l1) edge node [right] {$b$} (l2a)
  (l1) edge node [right] {$c$} (l3a)
  (l2a) edge node [right] {$d$}  (l3b)
  (l3a) edge node [right] {$e$}  (l4) edge node [right] {$f$}  (l2b)
  ;
\end{tikzpicture}}
\caption{PZG before merging}
\label{fig:PZGex1}
\end{subfigure}
\hfill
\begin{subfigure}[b]{0.35\textwidth}
\scalebox{0.8}{\begin{tikzpicture}[PZG, node distance=1.3cm]
\node[symbstate] (l0) {$\loc_0,True$};
\node[symbstate, below left of= l0] (l1) {$\loc_1,p\leq 1$};
\node[symbstate, accepting, below right of= l0] (l2) {$\loc_2,True$};
\node[symbstate, below right of=l1] (l3) {$\loc_3, True$};
\node[symbstate, below of=l3] (l4) {$\loc_4,p\leq 1$};
\path[->]
  (l0) edge node [left] {$a$} (l1) edge node [right] {$b$} (l2)
  (l1) edge node [left] {$c$} (l3)
  (l2) edge [bend right] node [right] {$d$} (l3)
  (l3) edge [bend right] node [right] {$f$} (l2) edge node [right] {$e$} (l4)
  ;
\end{tikzpicture}}
\caption{PZG after merging}
\label{fig:PZGex2}
\end{subfigure}
\caption{The original PZG satisfies $\mathbf{G}\, (\loc_2 \to \mathbf{G}\,\neg\loc_4)$, which
is violated after merging.
Similarly, merging introduces a spurious infinite loop containing $\loc_2$.}
\label{fig:ex}
\end{figure}

Note that path properties in~$\PZG$ are not always preserved in~$\PZG'$, as the following example shows. 
Also, liveness properties in~$\PZG$ (such as ``every path visits location $\loc$ infinitely often'')
are not necessarily preserved in~$\PZG'$.

\begin{example}
\cref{fig:PZGex1} shows the parametric zone graph of the PTA in \cref{fig:PTAex}.
The maximal paths are $\loc_0, \loc_1, \loc_3, \loc_4$ and $\loc_0, \loc_1, \loc_3, \loc_2$ 
(for $p\leq 1$) and $\loc_0, \loc_2, \loc_3$ (for $p>1$).
All maximal paths satisfy the LTL property $\mathbf{G}\, (\loc_2 \to \mathbf{G}\,\neg\loc_4)$ (``no $\loc_4$ after an $\loc_2$'').
Also, there is no loop (infinite run) containing $\loc_2$.
However, the result after merging in \cref{fig:PZGex2} introduces the spurious path
$\loc_0, \loc_2, \loc_3, \loc_4$, violating the first property. It also introduces a spurious
loop $\loc_0,(\loc_2,\loc_3)^\omega$, around $\loc_2$.
\end{example}
 
This example uses parameters, but no clocks. \cite[Fig.~4]{LODLV13} shows an example with only clocks
(\ie{} a timed automaton) where a spurious loop is introduced by zone inclusion (subsumption), which is
just a special case of zone merging.

\section{Experiments}\label{section:experiments}

We evaluate here the effect of the merging heuristics on reachability synthesis, \ie{} the synthesis of the parameter valuations for which a given reachability property holds.
The synthesis algorithm explores the PZG to find all valid parameter valuations.

We implemented all our heuristics in the \imitator{} parametric timed model checker~\cite{Andre21}.
The parametric zones in the symbolic states are encoded using polyhedra.
All operations on polyhedra, and notably the mergeability test, are performed using the Parma Polyhedra Library~\cite{BHZ08}.
We also reimplemented and compared with the original merging technique of \imitator{} 2.12, 
which was an upgrade of the merging technique (in \imitator{} 2.6.1) of~\cite{AFS13atva}.
\subsection{Dataset and Experimental Environment}\label{subsection:dataset}

We use the full set of models with reachability properties from the \imitator{} benchmark library~\cite{AMP21}.
The library is made of a set of \emph{benchmarks}. Each benchmark may have different \emph{models} and each model comes with one or more \emph{properties}.
For example, \styleBen{Gear} comes with ten models, of different sizes, named \styleBen{Gear-1000} to \styleBen{Gear-10000}; each of them may have one or more properties. 

Our dataset comprises 124 pairs made of a model and a reachability property (\ie{}
124 possible executions of \imitator{}).
We set a timeout of 120\,s; only 102 executions terminate within this time bound for at least one of the merging heuristics.
For 42 of these executions, at least one of the heuristics performs at least one successful merge. %
Full statistics on our dataset are given in~\cref{tab:dataset-size}.
\begin{table}[b!]
	\caption{Size of our dataset}
	\scriptsize
	\centering
	
	\begin{tabular}{|l|r|r|r|}
		\hline
		          \rowHeader{}            & \# benchmarks & \# models &    \# properties     \\
	Whole reachability dataset          &          49          &        84        &             124             \\ \hline
	Where at least one execution ends within $120$\,s &          35          &        68        &             102             \\ \hline
	Where at least one merge is performed   &          24          &        35        &             42              \\ \hline
	\end{tabular}
	
	\label{tab:dataset-size}
\end{table}

Experiments were run on an Intel Xeon Gold 5220 (Cascade Lake-SP, 2.20GHz, 1 CPU/node, 18 cores/CPU) with 96 GiB running Linux Ubuntu~20.\footnote{%
	We used \imitator{} \href{https://github.com/imitator-model-checker/imitator/releases/tag/v3.3.0-beta-2}{3.3-beta-2} ``Cheese Caramel au beurre salé''.
	Sources, binaries, models, raw results and full experiments tables are available at
	\href{https://www.doi.org/10.5281/zenodo.6806915}{\nolinkurl{10.5281/zenodo.6806915}}.
}
\subsection{Description of the Experiments}\label{subsection:experiments-description}

We compare each combination of the heuristics proposed in~\cref{subsec:options}. We reference each merge heuristic as a combination of three or four letters:
\begin{enumerate}
	\item \texttt{R} or \texttt{O}: the state-space is updated by reconstruction (\texttt{R}) or on-the-fly (\texttt{O}); %
	\item \texttt{V},\texttt{Q} or \texttt{O}: the selected candidates are \namevisited{} (\texttt{V}), \namequeue{} (\texttt{Q}) or \nameordered{} (\texttt{O}); %
	\item \texttt{M} or \texttt{C}: state-space is updated for each merge (\texttt{M}) or after all candidates (\texttt{C}); %
	\item \texttt{r}: the restart option is enabled (nothing otherwise). %
\end{enumerate}

These algorithms are compared according to:
\begin{ienumerate}
	\item the total computation time needed for a property; and
	\item the size of the generated state space.
\end{ienumerate}

Our results are obtained over the 102 executions of the dataset for which at least one algorithm ends before reaching the 120\,s timeout.
We do not use any penalty on executions that do not end: their execution time is set to the timeout (120\,s) in the subsequent analyses.
The metrics tagged by ``(merge)'' in \cref{tab:results-all-heuristic} are computed over the 42 executions
where some states can be merged,  while the ``(no merge)'' only consider the 60 executions where no merge can be made. %

\begin{table}[tb]	
	\centering
	\scriptsize
	
	\caption{Partial results comparing merge heuristics}
	\label{tab:results-all-heuristic}
	
	\begin{tabular}{|c|l|l|l|l|l|}
		\hline
		\rowHeader &&\nomerge&\mergeTwoTwelve&\mergeRVMr&\mergeOQM\\
		\hline
		\multirow{8}{*}{\rotatebox[origin=c]{90}{Time}}&\# wins&
			\cellThree{24}&20&\cellFive{22}&\cellOne{\cellBest{42}}\\
		&Avg (s)&10.0&\cellFive{5.47}&\cellThree{4.56}&\cellOne{\cellBest{3.77}}\\
		&Avg (merge) (s)&18.8&\cellFive{7.83}&\cellThree{5.57}&
			\cellOne{\cellBest{3.63}}\\
		&Avg (no merge) (s)&\cellThree{3.83}&\cellOne{\cellBest{3.82}}&
			\cellFive{3.85}&3.88\\
		&Median (s)&1.39&\cellFive{1.2}&\cellThree{1.14}&\cellOne{\cellBest{1.12}}\\
		&Norm.\ avg&1.0&\cellTwo{0.91}&\cellThree{0.91}&\cellOne{\cellBest{0.87}}\\
		&Norm.\ avg (merge)&1.0&\cellFive{0.75}&\cellThree{0.74}&
			\cellOne{\cellBest{0.64}}\\
		&Norm.\ avg (no merge)&\cellOne{\cellBest{1.0}}&\cellThree{1.02}&
			\cellFive{1.03}&\cellFive{1.03}\\
		\hline
		\multirow{5}{*}{\rotatebox[origin=c]{90}{States}}&\# wins&
			0&\cellThree{19}&\cellOne{\cellBest{37}}&\cellFive{16}\\
		&Avg&11443.08&\cellThree{11096.54}&\cellOne{\cellBest{11064.37}}&
			\cellFive{11120.79}\\
		&Avg (merge)&1512.02&\cellThree{670.43}&\cellOne{\cellBest{592.31}}&
			\cellFive{729.33}\\
		&Median&2389.5&\cellThree{703.5}&\cellOne{\cellBest{604.5}}&\cellFive{905.0}\\
		&Norm.\ avg&1.0&\cellThree{0.86}&\cellOne{\cellBest{0.84}}&\cellFive{0.88}\\
		\hline
	\end{tabular}
\end{table}

We present in~\cref{tab:results-all-heuristic} some of the experimental results obtained for the different merge heuristics that allow the best reduction of computation time or in the state-space size. The results for all the heuristics are presented in~\cref{appendix:results:results-all-heuristics}.
In order to allow a good visualization of the results, the best result in each cell is given in \textbf{bold}, while the level of green denotes the ``quality'' of the value in each cell (white is worst, and 100\,\% green is best).

The different lines tabulate the following information: %
\begin{ienumerate}
	\item the number of wins over the computation time, \ie{} the number of executions for which the current heuristics gives the smallest execution time;
	\item the average time (in s) over all executions;
	\item the average time (in s), excluding executions where no states can be merged for any heuristics;
	\item the average time (in s) for only the executions where no states can be merged for any heuristics;
	\item the median time (in s) over all executions;
	\item the normalized time average, compared to the \nomerge{} results, \ie{} the ratio between the heuristic execution time and the \nomerge{} one;
	\item the normalized time average, excluding executions where no states can be merged for any heuristics;
	\item the normalized time average, for only the executions where no states can be merged for any heuristics;
	\item the number of wins over the size of the state space (\ie{} the total number of symbolic states after merging);
	\item the average size of the state space over all models;
	\item the average size of the state space over all models, excluding the executions where no states can be merged;
	\item the median size of the state space over all models;
	\item the normalized size average, compared to the \nomerge{} results.
\end{ienumerate}

The reason to give both an average time (resp.\ number of states)
	and a normalized time (resp.\ number of states)
	is because both metrics complement each other:
the weight of the large models has a higher influence in the average (which can be seen as unfair, as a few models have a large influence), while all models have equal influence in the normalized average (which can also be seen as unfair, as very small models have the same influence as very large models).

In \cref{tab:results-all-heuristic} (and in~\cref{appendix:results:results-all-heuristics}), we notice that the best (\ie{} smallest) times are obtained when the merging is performed on the queue and when the update is done after a performed merge, even though doing it with a reconstruction of the state space after each step loses time compared to the \namefly{} heuristic.
Moreover, restarting when a merge is performed does not seem to bring any gain in time.
Thus, with respect to time, when the winner is \mergeOQM{} (\ie{} merging when the candidates are taken from the \namequeue{}, when the update is done on-the-fly after each merge without any restart), which minimizes both the time when a merge is possible, but also when considering models where no merge can be performed.
Moreover, this heuristic gives the smaller times for the executions where no merge can be done.

Concerning the state space size, the winner is \mergeRVMr{} (\ie{} merging when the candidates are taken from the \namevisited{} states, 
updating the state space by a reconstruction after each merge, and with a restart if a merge can be performed). This performs
more checks to identify states for merging (comparing with all the visited, not only those in the queue), thus reducing the state space even more. 

Note that the methods \nomerge{} and \mergeTwoTwelve{} are almost always the losers (\ie{} slowest and largest state space), except for the heuristic where the update of the state-space is performed after the list of candidates.
Concerning our new heuristics, we note that \mergeOQM{} decreases the average computation time to 69\,\% when compared to the previous merging heuristic (\mergeTwoTwelve{}~\cite{AFS13atva}), and even to only~46\,\% (\ie{} a division by a factor $> 2$) compared with \mergeTwoTwelve{} on the subset of models for which at least one merge can be done.
Compared to disabling merging (\nomerge{}), our new heuristic \mergeOQM{} decreases to 38\,\% on the whole benchmark set, and even to 19\,\% (\ie{} a division by a factor $> 5$) on the subset of models for which at least one merge can be done.
This leads us to consider the new combination of merging only in the queue and with an on-the-fly update after each merge and without restart (heuristic \mergeOQM{}) as the default merging heuristic in \imitator{}.
For use cases that require a minimal state space, the new combination \mergeRVMr{} is the recommended option.
Note that this version is still faster on average (83\,\%) than the previous heuristic (\mergeTwoTwelve),
and more than twice as fast (46\,\%) as not merging at all (\nomerge).

\section{Conclusion}\label{section:conclusion}

In this paper, we investigated the importance of the merging operations in reachability synthesis using parametric timed automata.
We investigated different combinations of options.
The chosen heuristic (\mergeOQM{}, when the candidates are taken from the \namequeue{}, when the update is done on-the-fly after each merge without any restart) brings a decrease to 38\,\% of the average computation time for our entire benchmarks library compared to the absence of merging.
Compared to the previous merging heuristic from~\cite{AFS13atva}, the gain of our new heuristic is a decrease to 69\,\% of the average computation time---meaning that our new heuristic decreases the computation time by~31\,\% compared to the former heuristic from~\cite{AFS13atva}.
In other words, despite the cost of the mergeability test, the overall gain is large and shows the importance of the merge operation for parameter synthesis.
We also provide a heuristic for use cases where a minimal state space is important, for instance for a follow-up analysis.
Even though this is not the fastest heuristic, it is still faster than not merging at all, and faster than the old merging heuristic~\cite{AFS13atva}.
Our experiments show the high importance of carefully choosing the merging heuristics.
Our heuristics preserve the correctness of parameter synthesis for reachability properties.

\paragraph{Future work.}
We noted that pruning merged states away (``garbage collection'') can prevent future merges.
Another option would be to keep such states in a collection of ``potential mergers''.
These extra states could be useful as ``glue'' to merge a number of other states, that otherwise could not be merged, into one superstate---but at the cost of more memory. Another option could be to
merge more than two states in one go.
These options remain to be investigated.

It is well-known that less heuristics can be used for \emph{liveness} properties than for reachability properties.
Investigating whether \emph{some} merging can still be used for liveness synthesis (\ie{} the synthesis of parameter valuations for which some location is infinitely often reachable) is an interesting future work.

Another, more theoretical question is to define and compute the ``best possible merge''. Currently, the result of merging is not canonical, since it depends on the exploration order 
and the order of searching for siblings. We have not found a candidate definition that
minimizes the state space and provides natural, canonical merge representatives.

Finally, investigating the recent PPLite~\cite{BZ18,BZ20} instead of PPL for polyhedra computation is on our agenda.

\section*{Acknowledgements}

We thank Benjamin Loillier for helping us testing our artifact.
Experiments presented in this paper were carried out using the Grid'5000 testbed, supported by a scientific interest group hosted by Inria and including CNRS, RENATER and several universities as well as other organizations (see \url{https://www.grid5000.fr}).

\newpage
\appendix
\section{Results for all Heuristics on the Full Benchmark}\label{appendix:results:results-all-heuristics}

\begin{center}
\rotatebox{-90}{%
\begin{minipage}{\textwidth}
	\scriptsize

\scalebox{.95}{
\begin{tabular}{|c|l|l|l|l|l|l|l|l|l|l|l|l|l|l|l|}
	\hline
	\rowHeader &With Restarting
&\nomerge&\mergeTwoTwelve&\mergeRVMr&\mergeRVCr&\mergeRQMr&\mergeRQCr&\mergeROMr&\mergeROCr&\mergeOVMr&\mergeOVCr&\mergeOQMr&\mergeOQCr&\mergeOOMr&\mergeOOCr\\
	\hline
	\multirow{8}{*}{\rotatebox[origin=c]{90}{Time}}&\# wins&
		\cellOne{\cellBest{20}}&\cellThree{17}&\cellFive{15}&2&10&3&5&3&9&4&8&8&4&8\\
	&Avg (s)&10.0&5.47&4.56&46.7&\cellFive{3.84}&43.53&4.69&49.05&5.58&5.7&
		\cellOne{\cellBest{3.79}}&\cellThree{3.81}&5.54&5.63\\
	&Avg (merge) (s)&18.8&7.83&5.57&17.8&\cellFive{3.83}&10.1&5.86&18.69&8.0&8.32&
		\cellOne{\cellBest{3.66}}&\cellThree{3.7}&7.89&8.14\\
	&Avg (no merge) (s)&\cellThree{3.83}&\cellOne{\cellBest{3.82}}&\cellFive{3.85}&
		66.93&\cellFive{3.85}&66.93&3.88&70.29&3.89&3.87&3.88&3.88&3.89&3.87\\
	&Median (s)&1.39&1.2&\cellFive{1.14}&4.85&\cellFive{1.14}&2.98&1.19&6.49&1.15&
		1.15&\cellThree{1.12}&\cellOne{\cellBest{1.11}}&1.15&1.17\\
	&Nrm.\ avg&1.0&\cellFive{0.91}&\cellFive{0.91}&9.06&\cellOne{\cellBest{0.87}}&
		8.91&0.92&10.34&0.92&0.92&\cellThree{0.88}&\cellOne{\cellBest{0.87}}&0.92&0.93\\
	&Nrm.\ avg (merge)&1.0&0.75&0.74&1.74&\cellThree{0.66}&1.4&0.76&1.86&
		\cellFive{0.75}&0.76&\cellThree{0.66}&\cellOne{\cellBest{0.65}}&0.77&0.78\\
	&Nrm.\ avg (no-mrg)&\cellOne{\cellBest{1.0}}&\cellThree{1.02}&\cellFive{1.03}&
		14.27&\cellThree{1.02}&14.25&\cellFive{1.03}&16.38&\cellFive{1.03}&
		\cellFive{1.03}&1.04&\cellFive{1.03}&\cellFive{1.03}&1.04\\
	\hline
	\multirow{5}{*}{\rotatebox[origin=c]{90}{States}}&\# wins&0&19&
		\cellOne{\cellBest{32}}&\cellFive{29}&15&15&\cellFive{29}&
		\cellThree{30}&20&20&16&16&20&20\\
	&Avg&11445.61&11096.54&\cellOne{\cellBest{11064.37}}&11106.09&11120.34&11120.55&
		\cellThree{11066.85}&11105.73&11089.5&11089.5&11118.73&11118.73&
		\cellFive{11087.77}&\cellFive{11087.77}\\
	&Avg (merge)&1518.17&670.43&\cellOne{\cellBest{592.31}}&693.62&728.24&728.74&
		\cellThree{598.33}&692.74&653.33&653.33&724.31&724.31&\cellFive{649.14}&
		\cellFive{649.14}\\
	&Median&2389.5&703.5&\cellOne{\cellBest{604.5}}&\cellThree{607.0}&905.0&905.0&
		\cellOne{\cellBest{604.5}}&\cellThree{607.0}&\cellFive{701.0}&\cellFive{701.0}&
		905.0&905.0&\cellFive{701.0}&\cellFive{701.0}\\
	&Nrm.\ avg&1.0&\cellFive{0.86}&\cellOne{\cellBest{0.84}}&\cellThree{0.85}&0.88&
		0.88&\cellOne{\cellBest{0.84}}&\cellThree{0.85}&\cellFive{0.86}&\cellFive{0.86}&
		0.88&0.88&\cellThree{0.85}&\cellThree{0.85}\\
	\hline
\end{tabular}
}
\medskip

\scalebox{.95}{
\begin{tabular}{|c|l|l|l|l|l|l|l|l|l|l|l|l|l|l|l|}
	\hline
	\rowHeader &No Restarting
&\nomerge&\mergeTwoTwelve&\mergeRVM&\mergeRVC&\mergeRQM&\mergeRQC&\mergeROM&\mergeROC&\mergeOVM&\mergeOVC&\mergeOQM&\mergeOQC&\mergeOOM&\mergeOOC\\
	\hline
	\multirow{8}{*}{\rotatebox[origin=c]{90}{Time}}&\# wins&
		\cellOne{\cellBest{20}}&\cellFive{17}&11&4&8&3&4&2&6&6&\cellThree{19}&6&8&8\\
	&Avg (s)&10.0&5.47&\cellFive{4.56}&46.48&3.85&43.31&4.65&48.92&5.09&5.19&
		\cellOne{\cellBest{3.77}}&\cellThree{3.78}&5.16&5.25\\
	&Avg (merge) (s)&18.8&7.83&5.58&17.19&\cellFive{3.86}&9.57&5.74&18.41&6.8&7.07&
		\cellOne{\cellBest{3.63}}&\cellThree{3.64}&7.0&7.24\\
	&Avg (no merge) (s)&\cellThree{3.83}&\cellOne{\cellBest{3.82}}&\cellFive{3.85}&
		66.97&\cellFive{3.85}&66.94&3.89&70.28&3.89&3.88&3.88&3.88&3.88&3.87\\
	&Median (s)&1.39&1.2&\cellFive{1.14}&4.34&\cellFive{1.14}&2.82&1.17&5.4&
		\cellTwo{1.13}&\cellOne{\cellBest{1.12}}&\cellOne{\cellBest{1.12}}&
		\cellOne{\cellBest{1.12}}&\cellTwo{1.13}&1.15\\
	&Nrm.\ avg &1.0&0.91&\cellFive{0.9}&9.04&\cellThree{0.88}&8.91&0.92&10.33&
		\cellFive{0.9}&0.91&\cellOne{\cellBest{0.87}}&\cellOne{\cellBest{0.87}}&
		0.91&0.92\\
	&Nrm.\ avg (merge)&1.0&0.75&0.74&1.69&\cellThree{0.65}&1.38&0.76&1.83&
		\cellFive{0.72}&0.74&\cellOne{\cellBest{0.64}}&\cellThree{0.65}&0.74&0.76\\
	&Nrm.\ avg (no-mrg)&\cellOne{\cellBest{1.0}}&\cellThree{1.02}&
		\cellThree{1.02}&14.28&1.04&14.27&\cellFive{1.03}&16.38&\cellFive{1.03}&
		\cellFive{1.03}&\cellFive{1.03}&\cellFive{1.03}&\cellFive{1.03}&1.04\\
	\hline
	\multirow{5}{*}{\rotatebox[origin=c]{90}{States}}&\# wins&
		0&19&23&\cellFive{26}&15&15&\cellThree{29}&\cellOne{\cellBest{32}}&
		19&19&15&15&19&19\\
	&Avg&11445.61&11096.54&\cellThree{11073.57}&11105.33&11123.2&11122.24&
		\cellOne{\cellBest{11068.0}}&11105.41&11096.89&11096.89&11120.79&11120.79&
		\cellFive{11090.01}&\cellFive{11090.01}\\
	&Avg (merge)&1518.17&670.43&\cellThree{614.64}&691.79&735.17&732.83&
		\cellOne{\cellBest{601.12}}&691.98&671.29&671.29&729.33&729.33&
		\cellFive{654.57}&\cellFive{654.57}\\
	&Median&2389.5&703.5&636.5&\cellThree{583.0}&905.0&905.0&\cellFive{604.5}&
		\cellOne{\cellBest{577.0}}&706.5&706.5&905.0&905.0&701.0&701.0\\
	&Nrm.\ avg&1.0&\cellFive{0.86}&\cellOne{\cellBest{0.84}}&\cellThree{0.85}&0.88&
		0.88&\cellOne{\cellBest{0.84}}&\cellThree{0.85}&\cellFive{0.86}&
		\cellFive{0.86}&0.88&0.88&\cellFive{0.86}&\cellFive{0.86}\\
	\hline
\end{tabular}
}
\end{minipage}
}
\end{center}

	\newcommand{\CCIS}{Communications in Computer and Information Science}
	\newcommand{\ENTCS}{Electronic Notes in Theoretical Computer Science}
	\newcommand{\FAC}{Formal Aspects of Computing}
	\newcommand{\FundInf}{Fundamenta Informaticae}
	\newcommand{\FMSD}{Formal Methods in System Design}
	\newcommand{\IJFCS}{International Journal of Foundations of Computer Science}
	\newcommand{\IJSSE}{International Journal of Secure Software Engineering}
	\newcommand{\IPL}{Information Processing Letters}
	\newcommand{\JAIR}{Journal of Artificial Intelligence Research}
	\newcommand{\JLAP}{Journal of Logic and Algebraic Programming}
	\newcommand{\JLAMP}{Journal of Logical and Algebraic Methods in Programming} %
	\newcommand{\JLC}{Journal of Logic and Computation}
	\newcommand{\LMCS}{Logical Methods in Computer Science}
	\newcommand{\LNCS}{Lecture Notes in Computer Science}
	\newcommand{\RESS}{Reliability Engineering \& System Safety}
	\newcommand{\STTT}{International Journal on Software Tools for Technology Transfer}
	\newcommand{\TCS}{Theoretical Computer Science}
	\newcommand{\ToPNoC}{Transactions on Petri Nets and Other Models of Concurrency}
	\newcommand{\TSE}{{IEEE} Transactions on Software Engineering}

\newpage

\renewcommand*{\bibfont}{\small}
\printbibliography[title={References}]

\end{document}